\documentclass[12pt]{article}

\usepackage[dvipsnames]{xcolor}
\usepackage{amsthm}
\usepackage{mathrsfs}
\usepackage{amssymb}
\usepackage{amsmath}
\usepackage{stmaryrd}
\usepackage{hyperref}
\usepackage{fullpage} 
\usepackage[left]{lineno}

\newtheorem{thm}{Theorem}

\newtheorem{lem}{Lemma}
\newtheorem{cor}{Corollary}

\newtheorem{prb}{Problem}

\usepackage{dsfont}  
\usepackage[noend]{algorithmic}

\usepackage[normalem]{ulem}
\usepackage{cancel}
\usepackage{enumitem}

\usepackage{todonotes}

\newcommand{\R}{\mathbb{R}}
\newcommand{\eps}{\varepsilon}
\newcommand{\CP}{\mathord{\it CP}}
\newcommand{\NW}{\mathord{\it NW}}

\DeclareMathOperator{\diam}{diam}
\newcommand{\DIAM}{\diam}
\DeclareMathOperator{\polylog}{polylog}

\title{Closest Pair Queries in Vertical Slabs and Tight 
Bounds on the Number of Possible Answers\footnote{A preliminary version of this paper appeared in the \emph{Proceedings of the 19th International Symposium on Algorithms and Data Structures (Toronto, ON, 2025)}, LIPIcs~349, 
\href{https://doi.org/10.4230/LIPIcs.WADS.2025.8}{doi:10.4230/LIPIcs.WADS.2025.8}.} 
}

\author{
Ahmad Biniaz\thanks{School of Computer Science, University of Windsor, 
     Windsor, Canada. Email: abiniaz@uwindsor.ca.
     Research supported by NSERC.} 
\and
Prosenjit Bose\thanks{School of Computer Science, Carleton University, 
      Ottawa, Canada. 
      Email: \{jit,anil,michiel\}@scs.carleton.ca. 
      Research supported by NSERC.} 
\and
Chaeyoon Chung\thanks{Department of Computer Science and Engineering, 
     Pohang University of Science and Technology, Pohang, Republic of 
     Korea. Email: chaeyoon17@postech.ac.kr. 
     Research supported by Institute of Information \& Communications 
     Technology Planning \& Evaluation (IITP) grant funded by the Korea 
     government (MSIT).} 
\and
Jean-Lou De Carufel\thanks{School of Electrical Engineering and Computer 
      Science, University of Ottawa, Ottawa, Canada. 
      Email: jdecaruf@uottawa.ca, saeedodak@gmail.com. Research supported 
      by NSERC.} 
\and 
John Iacono\thanks{Department of Computer Science, Universit{\'e} libre 
        de Bruxelles, Brussels, Belgium. 
        Email: john@johniacono.com. This work was supported by the 
        Fonds de la Recherche Scientifique-FNRS.} 
\and
Anil Maheshwari\footnotemark[3] 
\and 
Saeed Odak\footnotemark[5] 
\and 
Michiel Smid\footnotemark[3] 
\and 
Csaba D. T\'{o}th\thanks{Department of Mathematics, California State 
    University Northridge, Los Angeles, CA, and Department of Computer 
    Science, Tufts University, Medford, MA, USA. 
    Email: csaba.toth@csun.edu. Research supported in 
    part by the NSF award DMS-2154347.}
}
\date{\today} 

\begin{document} 

\maketitle 

\begin{abstract}
Let $S$ be a set of $n$ points in $\R^d$, where $d \geq 2$ is a constant,
and let $H_1,H_2,\ldots,H_{m+1}$ be a sequence of vertical hyperplanes 
that are sorted by their first coordinates, such that exactly 
$n/m$ points of $S$ are between any two successive hyperplanes. 
Let $A(S,m)$ be the set of different closest pairs in the 
${{m+1} \choose 2}$ vertical slabs that are bounded by $H_i$ and $H_j$, 
over all $1 \leq i < j \leq m+1$. We prove tight bounds for the largest
possible size of $A(S,m)$, over all point sets of size $n$, and for 
all values of $1 \leq m \leq n$.

As a result of these bounds, we obtain, for any constant $\eps>0$, 
a data structure of size $O(n)$,
such that for any vertical query slab $Q$, the closest pair in the 
set $Q \cap S$ can be reported in $O(n^{1/2+\eps})$ time. Prior to this 
work, no linear space data structure with sublinear query time was 
known.
\end{abstract}


\section{Introduction}
Throughout this paper, we consider point sets in $\R^d$, where the
dimension $d$ is an integer constant. For any real number $a$, we define
the \emph{vertical hyperplane} $H_{a}$ to be the set
\[ H_{a} = \{ (x_1,x_2,\ldots,x_d) \in \R^d : x_1 = a \} .
\]
Note that this is a hyperplane with normal vector $(1,0,0,\ldots,0)$. 
For any two real numbers $a$ and $b$ with $a<b$, we define the
\emph{vertical slab} $\llbracket H_a, H_b \rrbracket$ to be the set
\[ \llbracket H_a, H_b \rrbracket = 
    \{ (x_1,x_2,\ldots,x_d) \in \R^d : a \leq x_1 \leq b \} . 
\]
Let $S$ be a set of $n$ points in $\R^d$, in which no two points have
the same first coordinate and all ${n} \choose {2}$ pairwise Euclidean
distances are distinct.

For any two real numbers $a$ and $b$ with $a<b$, we define
$\CP(S,H_a,H_b)$ to be the closest-pair among all points in the set
$\llbracket H_a, H_b \rrbracket \cap S$, i.e., all points of $S$
that are in the vertical slab $\llbracket H_a, H_b \rrbracket$.
If $\llbracket H_a, H_b \rrbracket \cap S$ has size at most one, then
$\CP(S,H_a,H_b) = \infty$.

Clearly, there are $\Theta(n^2)$ 
combinatorially 
different\footnote{The slabs $\llbracket H_a, H_b \rrbracket$ and $\llbracket H_{a'}, H_{b'} \rrbracket$ are combinatorially different 
with respect to $S$ if their intersections with $S$ are different.}
sets of
the form $\llbracket H_a, H_b \rrbracket \cap S$.
Sharathkumar and Gupta~\cite{sharathkumar2007range} have shown that, for $d=2$, the size of the set
\[ \{ \CP(S,H_a,H_b) : a < b \} 
\]
is only $O(n \log n)$. That is, even though there are $\Theta(n^2)$
combinatorially different vertical slabs with respect to $S$, the number of different closest
pairs in these slabs is only $O(n \log n)$.

In this paper, we generalize this result to the case when the
dimension $d$ can be any constant and the slabs
$\llbracket H_a, H_b \rrbracket$ come from a restricted set.

Let $m$ be an integer with $1 \leq m \leq n$, and let
$a_1 < a_2 < \cdots < a_{m+1}$ be real numbers such that
for each $i = 1,2,\ldots,m$, there are 
exactly\footnote{In order to avoid floors and ceilings, we assume for simplicity that $n$ is a multiple of $m$.}
$n/m$ points of $S$ 
in the interior of the vertical slab
$\llbracket H_{a_i}, H_{a_{i+1}} \rrbracket$.
Observe that this implies that all points in $S$ are in the interior of the vertical slab $\llbracket H_{a_1}, H_{a_{m+1}} \rrbracket$.

We define
\[ A(S,m) = \{ \CP(S,H_{a_i},H_{a_j}) : 1 \leq i < j \leq m+1 \} . 
\]
In other words, we consider all ${m+1} \choose 2$ slabs bounded by vertical hyperplanes whose
first coordinates belong to $\{ a_1,a_2,\ldots,a_{m+1} \}$. For each such slab, we consider the closest pair among all points of $S$ inside the slab. 
The set $A(S,m)$ consists of all \emph{different} closest pairs.
Finally, we define
\[ f_d(n,m) = \max \{ |A(S,m)| : |S|=n \} . 
\]
Using this notation, Sharathkumar and Gupta~\cite{sharathkumar2007range} have shown that
$f_2(n,n) = O(n \log n)$.

In dimension $d=1$, it is easy to see that $f_1(n,m) = \Theta(m)$. Our main results are the
following tight bounds on $f_d(n,m)$, for any constant $d \geq 2$ and
any $m$ with $1 \leq m \leq n$:

\begin{thm}
\label{thm-main1}
Let $d \geq 2$ be a constant, and let $m$ and $n$ be integers such that
$1 \leq m \leq n$.
\begin{enumerate}
\item If $m = O(\sqrt{n})$, then $f_d(n,m) = \Theta(m^2)$.
\item If $m = \omega(\sqrt{n})$, then
$f_d(n,m) = \Theta(n \log (m^2/n))$.
\item In particular, if $m=n$, then $f_d(n,m) = \Theta(n \log n)$.
\end{enumerate}
\end{thm}

\subsection{Motivation and related work}

In the \emph{range closest pair problem}, we have to store a given set
$S$ of $n$ points in $\R^d$ in a data structure such that queries of
the following type can be answered: Given a query range $R$ in $\R^d$,
report the closest pair among all points in the set $R \cap S$.

Many results are known for different classes of query ranges. We 
mention some of the currently best data structures. 
Xue \emph{et al.}~\cite{DBLP:journals/dcg/XueLRJ22} present data 
structures for the case when $d=2$ and the query ranges are quadrants, 
halfplanes, or axes-parallel rectangles. Again for the case when $d=2$, 
data structures for query regions that are translates of a fixed shape 
are given by Xue \emph{et al.}~\cite{translate}. Some results in any 
constant dimension $d \geq 3$ are given by 
Chan \emph{et al.}~\cite{higherdimension}. 
Xue~\cite{sodaXue19} considers colored point sets, where the goal is 
to report the closest pair of points with different colors that are 
inside a query range. For constant dimension $d \geq 2$, 
\cite{sodaXue19} presents data structures for different types of 
query regions that report $(1+\varepsilon)$-approximations for the
closest pair with different colors.
References to many other
data structures can be found in 
\cite{higherdimension,translate,DBLP:journals/dcg/XueLRJ22}.

Most of the currently known data structures use super-linear space. 
To the best of our knowledge, linear-sized data structures are known 
only for the following classes of regions, all in dimension $d=2$: 
Quadrants and halfplanes~\cite{DBLP:journals/dcg/XueLRJ22}, and 
translates of a fixed polygon (possibly with holes)~\cite{translate}.
In all these three cases, the query time is $O(\log n)$. 

If each query range $R$ is a vertical slab
$\llbracket H_a, H_b \rrbracket$, we refer to the problem as the
\emph{vertical slab closest pair problem}. In dimension $d=1$, it is
easy to obtain a data structure of size $O(n)$ such that the closest
pair in any ``vertical slab'' (i.e., interval on the real line) can be
computed in $O(\log n)$ time. In dimension $d=2$, Sharathkumar and Gupta~\cite{sharathkumar2007range} gave a data structure of size $O(n \log^2 n)$ that allows
queries to be answered in $O(\log n)$ time.
Xue \emph{et al.}~\cite{DBLP:journals/dcg/XueLRJ22} improved the space bound to $O(n \log n)$,
while keeping a query time of $O(\log n)$.
Both these results use the fact that $f_2(n,n) = O(n \log n)$. 
In fact, both data structures explicitly store the set 
$\{ \CP(S,H_a,H_b) : a < b \}$, whose size is equal to $f_2(n,n)$ 
in the worst case. 

The starting point of our work was to design a data structure of size
$O(n)$ for vertical slab closest pair queries. This led us to
the problem of determining the asymptotic value of the 
function $f_d(n,m)$. Using our bounds in
Theorem~\ref{thm-main1}, we will obtain the following result.

\begin{thm}
\label{thm-main2}
Let $d \geq 2$ be an integer constant and let $\eps>0$ be a real constant.
For every set $S$ of $n$ points in $\R^d$, there exists a data structure
of size $O(n)$ that allows vertical slab closest pair queries to be
answered in $O(n^{1/2+\eps})$ time.
\end{thm}

Note that, prior to our work, no $O(n)$-space data structure with a 
query time of $o(n)$ was known for $d \geq 2$.
We remark that Theorem~\ref{thm-main2} follows from the 
first (i.e., ``easy'') case of Theorem~\ref{thm-main1}. Because of this, 
Section~\ref{data-structure} can be read independently of the rest of the 
paper.

Besides closest pair queries, other query problems have 
been considered. For a point set $S$ and axes-parallel query rectangles $R$,
\cite{BrassKSS13} shows how the area and perimeter of the convex hull, the width, and the radius of the smallest enclosing disk of the points in
$R \cap S$ can be computed; and \cite{DavoodiSW12} shows how the farthest 
pair in $R \cap S$ can be reported. If the points in $S$ have weights, 
\cite{RahulGJR11} shows how to report the $k$ largest weights among the 
points in $R \cap S$. Finally, \cite{BergGM18} stores a set of geometric 
objects, such that all pairs 
of objects that intersect inside $R$ can be reported.

\subparagraph{Organization.}
In Section~\ref{secupper}, we will present the upper bounds in 
Theorem~\ref{thm-main1} on $f_d(n,m)$. The corresponding lower 
bounds will be given in Section~\ref{seclower}. The data structure 
in Theorem~\ref{thm-main2} will be presented in 
Section~\ref{data-structure}. 
We conclude in Section~\ref{future} with some open problems. 

\subparagraph{Notation and Terminology.}
Throughout the rest of this paper, the notions of left and right in $\R^d$ will always refer to the ordering in the first coordinate. That is, if $p=(p_1,p_2,\ldots,p_d)$ and $q=(q_1,q_2,\ldots,q_d)$ are two points in $\R^d$ with $p_1 < q_1$, then we say that $p$ is to the \emph{left} of $q$, and $q$ is to the \emph{right} of $p$. For a vertical hyperplane $H_a$, we say that $p$ is to the \emph{left} of $H_a$ if $p_1 < a$. If $p_1 > a$, then $p$ is to the \emph{right} of $H_a$. 

The Euclidean distance between any two points $p$ and $q$ in $\R^d$ will be denoted by 
$|| p-q ||$. The length, or norm, of any vector $v$ will be denoted by $||v||$. 


\section{Upper bounds on $f_d(n, m)$}
\label{secupper}

Let $d \geq 2$ be a constant, let $m$ and $n$ be integers with
$1 \leq m \leq n$, and let $S$ be a set of $n$ points in $\R^d$. Let
$a_1 < a_2 < \cdots < a_{m+1}$ be real numbers such that for each 
$i = 1,2,\ldots,m$, there are exactly $n/m$ points in $S$ 
between the vertical hyperplanes $H_{a_i}$ and $H_{a_{i+1}}$. 

For any $m$, it is clear that $f_d(n,m) = O(m^2)$, because there are  ${m+1}\choose{2}$ vertical slabs of the form 
$\llbracket H_{a_i}, H_{a_j} \rrbracket$. Thus, the upper bound in Theorem~\ref{thm-main1} holds when $m = O(\sqrt{n})$. In the rest of this 
section, we assume that $m = \omega(\sqrt{n})$. 

The following lemma was
proved by Sharathkumar and Gupta~\cite{sharathkumar2007range} 
for the case when $d=2$. This lemma will be the key tool to prove our 
upper bound on $f_d(n,m)$. 

\begin{lem} \label{lem:gupta-higher-diminsions}
    Let $S$ be a set of $n$ points in $\R^d$ and let $\mathcal{C}\mathcal{P}$ be the set of segments corresponding to the elements of $A(S,n)$. That is, for each pair in $A(S,n)$, the set $\mathcal{C}\mathcal{P}$ contains the line segment connecting the two points in this pair. For any vertical hyperplane $H$, the number of elements of $\mathcal{C}\mathcal{P}$ that cross $H$ is $O(n)$.    
\end{lem}


We first present a proof of this lemma for the case when $d=2$. We 
believe that our proof is simpler than the one 
in~\cite{sharathkumar2007range}. Afterwards, we present a proof 
for any dimension $d \geq 2$. 

\subparagraph{Proof of Lemma~\ref{lem:gupta-higher-diminsions} when $d=2$.}
We write $\ell$ for the vertical
line. 
We define a graph, $G^+$, with vertex set $S$. Each segment of $\mathcal{C}\mathcal{P}$ with a positive slope represents an edge in the graph $G^+$. Let $F$ be the subgraph of $G^+$ induced by the segments of $\mathcal{CP}$ that cross $\ell$. We will show that $F$ does not contain a cycle.

Suppose, to the contrary, that there is a cycle $C$ in $F$. Let $a$ and $b$ be the endpoints of the shortest edge in $C$ such that $a$ is to the left of $\ell$ and $b$ is to the right of $\ell$. Let $ac$ and $bd$ be the other edges of the cycle that are incident to $a$ and $b$, respectively. Since both $ab$
and $ac$ represent pairs in $\mathcal{CP}$ and both have a positive slope, 
we have $a_x < c_x < b_x$ and $a_y < b_y < c_y$. Similarly, we have 
$a_x < d_x < b_x$ and $d_y < a_y < b_y$; see Figure~\ref{figs:gupta}.

\begin{figure}
    \begin{center}
        \includegraphics[scale=0.75]{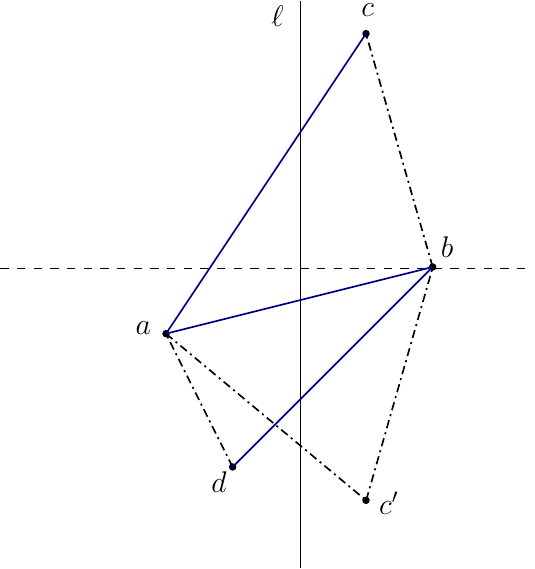}
    \end{center}
    \caption{The pairs in $A(S,n)$ with positive slope that cross $\ell$  do not contain a cycle.}
    \label{figs:gupta}
\end{figure}

Let $c'$ be the reflection of the point $c$ with respect to the 
horizontal line through $b$. Note that $||b-c'|| = ||b-c|| > ||b-d||$, because  
$bd$ represents a pair in $A(S,n)$ and the vertical slab 
$\llbracket b_x, d_x \rrbracket$ contains the point $c$. Since 
$||b-c'|| > ||b-d||$, we have $c'_y < d_y$. We also have $d_y < a_y$ and $a_x < d_x < c_x = c'_x$. It follows that $||a-d|| < ||a-c'||$.

Consider the bisector of the segment $cc'$ (which is the horizontal line through $b$). Observe that the point $a$ is located on the same side as $c'$ with respect to this bisector. Therefore, 
$||a-c'|| < ||a-c||$. Combined with $||a-d|| < ||a-c'||$, this implies that 
$||a-d|| < ||a-c||$. This contradicts the facts that $ac$ represents a pair in $A(S,n)$ and the point $d$ is  in the slab $\llbracket a_x, c_x \rrbracket$.

A similar argument shows that the segments in $\mathcal{C}\mathcal{P}$ 
that cross $\ell$ and have negative slopes do not contain a cycle. Therefore, the total number of segments in $\mathcal{C}\mathcal{P}$  
that cross the line $\ell$ is $O(n)$.
\qed
\vspace{0.4cm}

To prove Lemma~\ref{lem:gupta-higher-diminsions} for dimensions 
$d \geq 2$, we will use the \emph{Well-Separated Pair Decomposition (WSPD)}, as introduced by Callahan and Kosaraju~\cite{well-separated-2}. 
Let $S$ be a set of $n$ points in $\R^d$ and let $s>1$ be a real number, called the \emph{separation ratio}. A WSPD for $S$ is a set of pairs 
$\{A_i, B_i\}$, for $i=1,2,\ldots,k$, for some positive integer $k$, such that 

\begin{enumerate} 
\item for each $i$, $A_i \subseteq S$ and $B_i \subseteq S$, 
\item for each $i$, there exist two balls $D$ and $D'$ of the same radius, say $\rho$, such that $A_i \subseteq D$, $B_i \subseteq D'$, and the distance between $D$ and $D'$ is at least $s \cdot \rho$, 
i.e., the distance between their centers is at least $(s+2) \cdot \rho$, 
\item for any two distinct points $p$ and $q$ in $S$, there is a unique index $i$ such that $p \in A_i$ and $q \in B_i$ or vice-versa. 
\end{enumerate} 

Consider a pair $\{A_i,B_i\}$ in a WSPD. If $p$ and $p'$ are two points in $A_i$ and $q$ is a point in $B_i$, then it is easy to see that 
\begin{equation}\label{eq:WSPD}
||p-p'|| \leq (2/s) \cdot ||p-q|| . 
\end{equation}


\begin{lem}[Callahan and Kosaraju~\cite{well-separated-2}] \label{callahan}
    Let $S$ be a set of $n$ points in $\R^d$, and let $s > 1$ be a real number. A well-separated pair decomposition for $S$, with separation ratio $s$, consisting of $O(s^d n)$ pairs, can be computed in $O(n\log{n} + s^d n)$ time.
\end{lem}

\noindent 
{\bf Proof of Lemma~\ref{lem:gupta-higher-diminsions}
when $d \geq 2$.}
Let $s>2$ be a constant and consider a WSPD $\{A_i,B_i\}$, $i=1,2,\ldots,k$, for the point set $S$ with separation ratio $s$, where $k=O(n)$; see Lemma~\ref{callahan}. We define the following geometric graph $G$ on the point set $S$. For each $i$ with $1 \le i \le k$, let 

\begin{itemize}
    \item $a_i$ be the rightmost point in $A_i$ that is to the left of $H$,
    \item $b_i$ be the leftmost point in $B_i$ that is to the right of $H$,
    \item $a'_i$ be the leftmost point in $A_i$ that is to the right of $H$, and
    \item $b'_i$ be the rightmost point in $B_i$ that is to the left of $H$. 
\end{itemize}
We add the edges $a_i b_i$ and $a'_i b'_i$ to the graph $G$. Note that some of these points may not exist, in which case we ignore the corresponding edge. It is clear that $G$ has $O(n)$ edges. The lemma will follow from the fact that every segment in $\mathcal{C}\mathcal{P}$ that crosses $H$ is an edge in $G$. 

Let $pq$ be a pair in $\mathcal{C}\mathcal{P}$ that crosses $H$, and let $Q$ be a vertical slab such that $pq$ is the closest pair in $Q \cap S$. We may assume, without loss of generality, that $p$ is to the left of $H$ and $q$ is to the right of $H$. Let $i$ be the index such that (i) $p \in A_i$ and $q \in B_i$ or (ii) $p \in B_i$ and $q \in A_i$. We may assume, without loss of generality, that (i) holds. 

We claim that $p = a_i$. To prove this, suppose that $p \neq a_i$. Then, since $p$ is to the left of $a_i$, $a_i$ is in the slab $Q$. Since $s>2$, Equation~\eqref{eq:WSPD} yields 
$||p - a_i|| < ||p-q||$, which is a contradiction. 
By a symmetric argument, we have $q = b_i$. Thus, $pq$ is an edge in $G$.
\qed
\vspace{0.4cm}


Lemma~\ref{lem:gupta-higher-diminsions} gives us a 
divide-and-conquer approach to prove an upper bound on $f_d(n,m)$: 

\begin{thm} \label{thm:odak-set}
Let $d \geq 2$ be a constant, and let $m$ and $n$ be integers with 
$m = \omega(\sqrt{n})$ and $m \leq n$. Then 
$f_d(n,m) = O(n\log{(m^2/n)})$. 
\end{thm}

\begin{proof}
    Let $S$ be a set of $n$ points in $\R^d$ for which 
    $f_d(n,m) = |A(S,m)|$. Let $a_1 < a_2 < \cdots < a_{m+1}$ be real numbers such that for each $i = 1,2,\ldots,m$, there are exactly $n/m$ points in $S$ that are strictly inside the vertical slab $\llbracket H_{a_i}, H_{a_{i+1}} \rrbracket$. 
    
    Let $H = H_{a_{1+m/2}}$. Observe that $n/2$ points of $S$ are to the left of $H$ and $n/2$ points of 
    $S$ are to the right of $H$. Denote these two subsets by $S^-$ 
    and $S^+$, respectively. Each pair in $A(S,m)$ is either a pair in 
    $A(S^-,m/2)$ or a pair in $A(S^+,m/2)$ or it crosses $H$. 
    Since $A(S,m)$ is a subset of $A(S,n)$, it follows from 
    Lemma~\ref{lem:gupta-higher-diminsions} that the number of pairs in 
    $A(S,m)$ that cross $H$ is $O(n)$. Thus,
     \begin{eqnarray*} 
         f_d(n,m) & = & |A(S,m)| \\ 
           & = & |A(S^-,m/2)| + |A(S^+,m/2)| + O(n) \\ 
           & \leq & 2 \cdot f_d(n/2,m/2) + O(n) . 
     \end{eqnarray*} 
    If we apply this recurrence $k$ times, we get 
    \[ f_d(n,m) \leq 2^k \cdot f_d (n/2^k,m/2^k) + O(kn) . 
    \]
    For $k=\log (m^2/n)$, we have $n/2^k = n^2/m^2$ and 
    $m/2^k = n/m$. Thus, 
    \[ f_d(n,m) \leq
         \frac{m^2}{n} \cdot f_d \left( n^2/m^2 , n/m \right) + 
         O( n \log(m^2/n) ) . 
    \]
Since $f_d \left( n^2/m^2 , n/m \right) = O(n^2/m^2)$, we conclude that 
\[ f_d(n,m) = O( n + n \log(m^2/n) ) = O( n \log(m^2/n)) .
    \qedhere
\]
\end{proof}

\section{Lower bounds on $f_d(n, m)$}
\label{seclower}

In this section, we prove the lower bounds on $f_d(n,m)$ as stated in 
Theorem~\ref{thm-main1}. We will prove these lower bounds for the case 
when $d=2$. It is clear that this will imply the same lower bound for 
any constant dimension $d \geq 2$. 

\subsection{The case when $m \leq 3 \sqrt{n}$}

\begin{thm} \label{thm:smid-set}
    Let $n$ and $m$ be positive integers with $n \geq m(m+1)$. 
    Then $f_2(n, m) = {{m+1} \choose 2}$.
\end{thm}

\begin{proof}
It is clear that $f_2(n, m) \leq {{m+1} \choose 2}$. To prove the lower
bound, we will construct a set $S$ of $n$ points in $\R^2$ such that 
the ${m+1} \choose 2$ vertical slabs have distinct closest pairs. 

For $i=1,2,\ldots,m+1$, we take $a_i = i$ and consider the corresponding 
hyperplane $H_i$. Let 
$\mathcal{Q} = \{ \llbracket H_i, H_j \rrbracket : 1 \leq i < j \leq m+1 \}$ be the set of all possible vertical slabs. We define the \emph{size} 
of a slab $\llbracket H_i, H_j \rrbracket$ to be the difference $j-i$ 
of their indices. 

We start by constructing a set $P$ of $m(m+1)$ points such that 
the slabs in $\mathcal{Q}$ contain distinct closest pairs in $P$, 
and for each $i=1,2,\ldots,m$, the slab $\llbracket H_i,H_{i+1} \rrbracket$
contains exactly $m+1$ points of $P$. 

Note that the slab $\llbracket H_1 , H_{m+1} \rrbracket$ has the largest
size. Let $p$ be an arbitrary point in $\llbracket H_1,H_2 \rrbracket$
and let $q$ be an arbitrary point in $\llbracket H_m,H_{m+1} \rrbracket$. 
We initialize $P = \{p,q\}$, $D = || p-q ||$, and delete the slab 
$\llbracket H_1 , H_{m+1} \rrbracket$ from $\mathcal{Q}$. 

As long as $\mathcal{Q}$ is non-empty, we do the following: 

\begin{itemize}
\item Take a slab $\llbracket H_i,H_j \rrbracket$ of largest size 
in $\mathcal{Q}$. 
\item Let $p$ be an arbitrary point in 
$\llbracket H_i,H_{i+1} \rrbracket$ such that $p$ is above the bounding 
box of $P$, and the distance between $p$ and any point in $P$ is more than $D+2$. 
\item Let $q$ be an arbitrary point in 
$\llbracket H_{j-1},H_j \rrbracket$ such that $q$ is above the bounding 
box of $P$, the distance between $q$ 
and any point in $P$ is more than $D+2$, and $|| p-q || = D+1$. 
\item Add $p$ and $q$ to $P$.
\item Set $D = || p-q ||$. 
\item Delete the slab $\llbracket H_i , H_j \rrbracket$ from $\mathcal{Q}$.
\end{itemize}

It is not difficult to see that the final point set $P$ has the 
properties stated above. 

To obtain the final point set $S$, of size $n$, we define a set $P'$ of 
$n-m(m+1)$ points, such that each point in $P'$ has distance more than 
$D$ to all points of $P$, the closest pair distance in $P'$ is more 
than $D$, and for each $i=1,2,\ldots,m$, the slab 
$\llbracket H_i , H_{i+1} \rrbracket$ contains $n/m - (m-1)$ points 
of $P'$. The point set $S = P \cup P'$ has the property that 
$|A(S,m)| = {{m+1} \choose 2}$. 
\end{proof}

\begin{cor} \label{cor-sqrt}
    Let $n$ and $m$ be sufficiently large positive integers with $n < m(m+1)$ and 
    $m \leq 3 \sqrt{n}$. Then $f_2(n,m) = \Omega(m^2)$. 
\end{cor}
\begin{proof}
For $i=1,2,\ldots,m+1$, we take $a_i = i$ and consider the 
corresponding hyperplane $H_i$. 

Let $m' = \sqrt{n}/4$ and $n' = m' (m' + 1)$. We apply 
Theorem~\ref{thm:smid-set}, where we replace $n$ by $n'$ 
and $m$ by $m'$. This gives us a set $S'$ of $n'$ points
with $|A(S',m')| = f_2(n',m')$. The points of $S'$ 
are between the hyperplanes $H_1$ and $H_{m'+1}$; for each 
$i =1,2,\ldots,m'$, the vertical slab
$\llbracket H_i , H_{i+1} \rrbracket$ contains $n'/m'$ 
points of $S'$. Note that 
\[ |A(S',m')| = {{m'+1} \choose 2} = \Omega\left( (m')^2 \right) .
\]
Let $D$ be the diameter of $S'$. 
Let $S$ be the union of $S'$ and a set of $n-n'$ 
additional points that have pairwise distances more than 
$D$, whose distances to the points in $S'$ are more 
than $D$, and such that for each $i=1,2,\ldots,m$, the 
vertical slab $\llbracket H_i , H_{i+1} \rrbracket$ 
contains $n/m$ points of $S$. It is clear that 
\[ f_2(n,m) \geq |A(S',m')|  = \Omega((m')^2) . 
\]
Note that this construction is possible, because 
(i) $n' < n$, (ii) $m' < m$, and (iii) $n'/m' < n/m$; 
these inequalities follow by straightforward algebraic 
manipulations, using the assumptions on $n$ and $m$ in the 
statement of the corollary. Finally, these assumptions 
imply that $m' \geq m/12$. We conclude that 
$f_2(n,m) = \Omega (m^2)$. 
\end{proof}

\subsection{Warm up: The case when $m=n$}
Theorem~\ref{thm:smid-set} and 
Corollary~\ref{cor-sqrt} proved Theorem~\ref{thm-main1} 
for the case when $m \leq 3 \sqrt{n}$.
Before we prove the lower bound for the remaining case, i.e., $m > 3 \sqrt{n}$, we 
consider the case when $m=n$, which will serve as a warm up.

\begin{thm} \label{thm:odak-set-lower}
    We have $f_2(n,n) = \Omega(n\log{n})$.
\end{thm}


\begin{proof}
We assume for simplicity that $n$ is a sufficiently large power of two. 
    We will construct a point set $S$ of size $n$ for which 
    $|A(S,n)| = \Omega(n \log n)$. 

    \paragraph{Construction of the point set $S$:}
Let $k = \log n$. For $i=0,1,\ldots,k-1$, let $x_i = 2^i$ and let
    $v_i = (x_i,y_i)$ be a vector, where the value of $y_i$ is inductively defined as follows: We set $y_{k-1} = 0$. Assuming that 
    $y_{k-1},y_{k-2},\ldots,y_{i+1}$ have been defined, we set $y_i$ to 
    an integer such that 
    \begin{equation} 
    \label{blacklozenge}
      ||v_i|| > 2\sum_{j=i+1}^{k-1} ||v_j|| . 
    \end{equation} 
    We define 
    \[ S = \left\{ \sum_{i=0}^{k-1} \beta_i v_i : 
        (\beta_0,\beta_1, \cdots, \beta_{k-1}) \in \{0,1\}^k
           \right\} .
    \] 

To give an example, let $n=8$ and, thus, $k=3$. We can take
$v_0 = (1,25)$, $v_1 = (2,8)$, and $v_2 = (4,0)$.
The coordinates of the points in the set $S$ are listed below, and a 
visualization of $S$ is in Figure~\ref{figs:example}.

\begin{center}
\begin{tabular}{|c|c|c|c|}
\hline
$\beta_0$ & $\beta_1$ & $\beta_2$ &
$\beta_0 v_0 + \beta_1 v_1 + \beta_2 v_2$ \\
\hline
\hline
$0$ & $0$ & $0$ & $(0,0)$ \\
\hline
$1$ & $0$ & $0$ & $(1,25)$ \\
\hline 
$0$ & $1$ & $0$ & $(2,8)$ \\
\hline 
$1$ & $1$ & $0$ & $(3,33)$ \\
\hline 
$0$ & $0$ & $1$ & $(4,0)$ \\
\hline
$1$ & $0$ & $1$ & $(5,25)$ \\
\hline 
$0$ & $1$ & $1$ & $(6,8)$ \\
\hline 
$1$ & $1$ & $1$ & $(7,33)$ \\
\hline
\end{tabular}
\end{center}

\begin{figure}
    \begin{center}
        \includegraphics[scale=0.75]{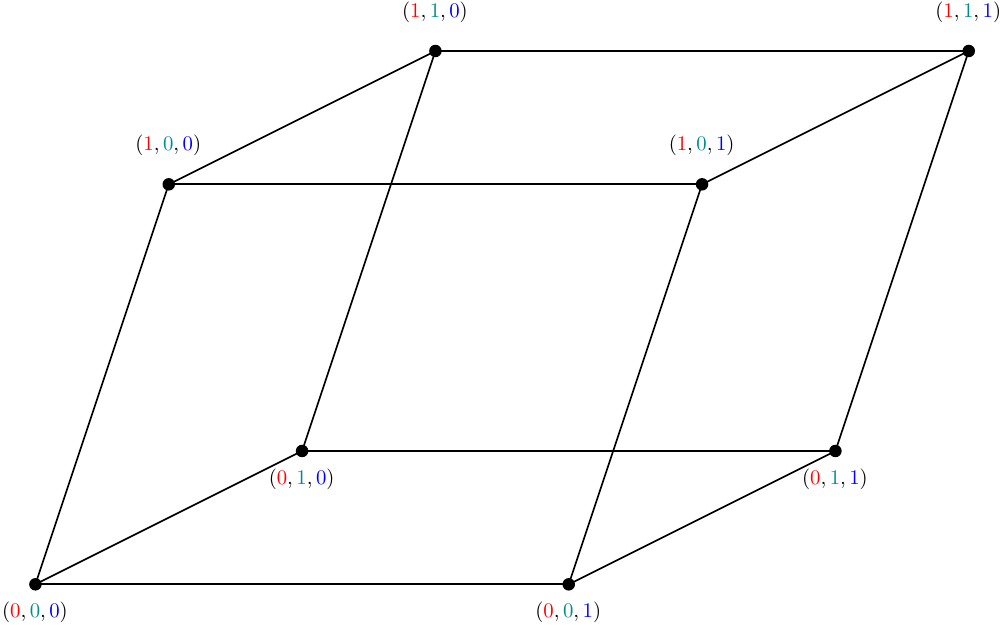}
    \end{center}
    \caption{An illustration of the point set $S$, together with the 
    corresponding $3$-dimensional hypercube $Q_3$.}
    \label{figs:example}
\end{figure}

    Note that each binary sequence of length $k$ represents a unique point 
    in $S$. Using this representation, each point of $S$ corresponds 
    to a vertex of a $k$-dimensional hypercube $Q_k$. 
    We will prove below that each edge of $Q_k$ corresponds to a closest
    pair in a unique vertical slab. Since $Q_k$ has 
    $k \cdot 2^{k-1} = \Omega(n \log n)$ edges, this will complete the 
    proof. 

\paragraph{Each edge of $Q_k$ corresponds to a closest pair:}
    Consider an arbitrary edge of $Q_k$. The two vertices of this edge 
    are binary sequences of length $k$ that have Hamming distance one, 
    i.e., they differ in exactly one bit. Let $t$ be the position at 
    which they differ. Observe that $0 \leq t \leq k-1$. Let $r$ and 
    $s$ be the points of $S$ that correspond to the two vertices of 
    this edge. Then $v_t$ is either $r-s$ or $s-r$. We will prove that 
    $r$ and $s$ form the closest pair in the vertical slab 
    $\llbracket H_{r_1} , H_{s_1} \rrbracket$, where $r_1$ and $s_1$ are the first 
    coordinates of $r$ and $s$, respectively (assuming that 
    $r_1 < s_1$). Note that $r_1$ and $s_1$ are integers. 

    Let $p$ and $q$ be two points in 
    $\llbracket H_{ r_1 } , H_{ s_1 }\rrbracket \cap S$ such that 
    $\{ p,q \} \neq \{ r,s \}$. We have to show that 
    $|| r-s || < || p-q ||$. Since $p$ and $q$ are points in $S$, 
    we can write them as 
    \[ p = \sum_{i=0}^{k-1} \beta_{p,i} v_i 
    \]
    and 
    \[ 
    q = \sum_{i=0}^{k-1} \beta_{q,i} v_i . 
    \]
    Let $\ell$ be the smallest index for which 
    $\beta_{p,\ell} \neq \beta_{q,\ell}$. Then 
    \[ || p-q || = \left| \left| \sum_{i=\ell}^{k-1} 
               ( \beta_{p,i} - \beta_{q,i} ) v_i \right| \right| . 
    \]
    By the triangle inequality, we have $||a+b|| \leq ||a||+||b||$ and 
    $||a+b|| \geq ||a||-||b||$ for any two vectors $a$ and $b$. These two inequalities, 
    together with (\ref{blacklozenge}), imply that 
    \[ || p-q || \geq || v_{\ell} || - 
            \sum_{i=\ell+1}^{k-1} || v_i || > 
            \frac{1}{2} \cdot || v_{\ell} || . 
    \]
    Thus, it suffices to show that 
    $|| v_{\ell} || \geq 2 \cdot || r-s || = 2 \cdot || v_t ||$.
    If we can show that $\ell < t$, then, using (\ref{blacklozenge}), 
    \[ || v_{\ell} || > 2 \sum_{i=\ell+1}^{k-1} || v_i || 
         \geq 2 \cdot || v_t || 
    \]
    and the proof is complete. 

    \paragraph{Completing the proof:}

Observe that all points of $S$ have pairwise distinct first coordinates. Also, the points $r$ and $s$ are on the boundary of
$\llbracket H_{r_1} , H_{s_1} \rrbracket$. Our assumption that 
$\{p,q\} \neq \{r,s\}$ implies that the 
horizontal distance between $p$ and $q$ (i.e., $| p_1-q_1 |$) 
    is smaller than the horizontal distance between $r$ and $s$.
    Recall that $v_t$ is either $r-s$ or $s-r$. Thus, the horizontal 
    distance between $r$ and $s$ is equal to the first coordinate of 
    the vector $v_t$, which is $x_t$. 

    Let $\ell'$ be the largest index for which 
    $\beta_{p,\ell'} \neq \beta_{q,\ell'}$. 
    Note that $\ell \leq \ell'$. 
    If $\ell = \ell'$, then 
    \[ p-q = \left( \beta_{p,\ell} - \beta_{q,\ell} \right) v_{\ell} 
    \]
    and 
    \[ | p_1 - q_1 | = | \beta_{p,\ell} - \beta_{q,\ell} | x_{\ell} 
               = x_{\ell} .
    \]
    Now assume that $\ell < \ell'$. We may assume, without loss
    of generality, that 
    $\beta_{p,\ell'} = 1$ and $\beta_{q,\ell'} = 0$. 
    Then (recall that $x_i = 2^i$), 
    \begin{eqnarray*} 
      | p_1 - q_1 | & = & \left| x_{\ell'} + \sum_{i=\ell}^{\ell'-1} 
               ( \beta_{p,i} - \beta_{q,i} ) x_i \right| \\ 
               & = & \left| 2^{\ell'} + \sum_{i=\ell}^{\ell'-1} 
               ( \beta_{p,i} - \beta_{q,i} ) 2^i \right| \\
        &\geq & 2^{\ell'} - \sum_{i=\ell}^{\ell'-1} 2^i  \\
          & = & 2^{\ell} \\
          & = & x_{\ell}. 
    \end{eqnarray*}
    We conclude that 
    $x_{\ell} \leq | p_1 - q_1 | < x_t$,
    which is equivalent to $\ell < t$. 
\end{proof}

\subsection{The case when $m > 3 \sqrt{n}$}

\begin{thm} \label{prop:toth-set}
    Let $n$ and $m$ be positive integers with $3 \sqrt{n} < m \leq n$.
    Then $f_2(n,m) = \Omega(n\log{(m^2/n)})$.
\end{thm}

\begin{proof}
Our approach will be to use multiple scaled and shifted copies of the construction in Theorem~\ref{thm:odak-set-lower} to define a set $S$ of 
$n$ points in $\R^2$ for which $|A(S,m)| = \Omega(n\log{(m^2/{n})})$.

For $i=1,2,\ldots,m+1$, we take $a_i = i$ and consider the corresponding 
vertical hyperplane $H_i$. For each $i=1,2,\ldots,m$, the point set $S$ will 
contain exactly $n/m$ points in the vertical slab
$\llbracket H_i , H_{i+1} \rrbracket$. 

\paragraph{Hypercube-sets:}
Throughout the proof, we will use the following notation. 
Let $v_0, v_1, \ldots, v_{k-1}$ be a sequence of pairwise 
distinct vectors in 
the plane. The \emph{hypercube-set} defined by these vectors is the 
point set 
\[ Q(v_0,v_1,\ldots,v_{k-1}) = 
    \left\{\sum_{i=0}^{k-1} \beta_i v_i 
        : (\beta_0, \beta_1, \cdots, \beta_{k-1}) \in \{0,1\}^k \right\}.
\]

For each $g = 1,2,\ldots,n/m$, we define a hypercube-set $Q_g$: 
\begin{itemize}
    \item Let $k_g = \lfloor \log (m/(2g-1)) \rfloor$.
    \item For each $i=0,1,\ldots,k_g-1$, let 
          $x_{g,i} = (2 g - 1) \cdot 2^i$
          and let 
          $v_{g,i} = ( x_{g,i} , y_{g,i} )$
          be a vector, whose second coordinate $y_{g,i}$ will be 
          defined later. 
    \item Let $Q_g = Q( v_{g,0} , v_{g,1} , \ldots, v_{g,k_g-1})$.
\end{itemize}
Since, for integers $g$, $g'$, $i$, and $i'$, 
$(2g-1) \cdot 2^i = (2g'-1) \cdot 2^{i'}$ if and only if 
$g=g'$ and $i = i'$, then all values $x_{g,i}$ are pairwise distinct. 

To define the values $y_{g,i}$, we sort all vectors $v_{g,i}$ by their 
first coordinates. We go through the sorted sequence in decreasing 
order:
\begin{itemize}
    \item For the vector with the largest first coordinate, we set its 
    $y$-value to zero. 
    \item For each subsequent vector $v_{g,i}$, we set $y_{g,i}$ to be an 
    integer such that 
    \begin{equation} \label{eq2} 
      || v_{g,i} || > 2 \sum_{g',i'} || v_{g',i'} || ,        
    \end{equation}
    where the summation is over all pairs $g',i'$ for which $y_{g',i'}$ 
    has already been defined (i.e., $x_{g,i} < x_{g',i'}$). 
\end{itemize}

\paragraph{Using the hypercube-sets to define a preliminary point set $S'$:}

We choose pairwise distinct real numbers $0 < \eps_g < 1$, for 
$g = 1,2,\ldots,n/m$, and set 
\[ \Delta = 1 + \max \{ \DIAM (Q_g) : 1 \leq g \leq n/m \},
\]
where $\DIAM(Q_g)$ denotes the diameter of the point set $Q_g$. 

For each $g=1,2,\ldots,n/m$ and $i=1,2,\ldots,2g-1$, let 
\[ S_{g,i} = Q_g + 
    ( i + \eps_g , 2 ( (g-1)^2 + i-1 ) \Delta ) ,\]
that is, $S_{g,i}$ is the translate of $Q_g$ by the vector $( i + \eps_g , 2 ( (g-1)^2 + i-1 ) \Delta )$.
We define $S'$ to be the union of all these sets $S_{g,i}$, i.e., 
\[ S' = \bigcup_{g=1}^{n/m} \bigcup_{i=1}^{2g-1} S_{g,i} . 
\]
Note that the sets $S_{g,i}$ are pairwise  disjoint: Indeed if $g\neq g'$ or $i\neq i'$, then the $y$-projections of $S_{g,i}$ and $S_{g',i'}$  (i.e., the sets of second coordinates) are disjoint by construction. 
Consequently, the size of the union of these point sets satisfies 
\[  |S'| 
= \sum_{g=1}^{n/m} \sum_{i=1}^{2g-1} |S_{g,i}|
= \sum_{g=1}^{n/m} \sum_{i=1}^{2g-1} 2^{k_g} 
      \leq \sum_{g=1}^{n/m} (2g-1) \cdot \frac{m}{2g-1} 
     =  n .
\]
For each $1 \le g \le n/m$, by construction of $Q_g$ and the fact that the sets $S_{g,i}$ are disjoint translations of $Q_g$, each slab $\llbracket H_j , H_{j+1} \rrbracket$ contains at most one point of $\bigcup_{i=1}^{2g-1} S_{g,i}$. Therefore, each slab $\llbracket H_j , H_{j+1} \rrbracket$ contains at most $n/m$ points of $S'$. 

\paragraph{The final point set $S$:}
To obtain the final point set $S$ of size $n$, we add $n-|S'|$ points to 
$S'$ such that 
each slab $\llbracket H_j , H_{j+1} \rrbracket$ contains exactly 
$n/m$ points of $S$, and the added points are sufficiently far from each 
other and from all points of $S'$.

\paragraph{The main claim implying correctness:}

In the rest of this proof, we will prove the following claim: For each 
$g=1,2,\ldots,n/m$, consider two binary strings of length $k_g$ that 
differ in exactly one position (recall that the number of such pairs of 
strings is equal to $k_g \cdot 2^{k_g-1}$). These strings correspond to
two points of the hypercube-set $Q_g$. Thus, for any 
$i=1,2,\ldots,2g-1$, they
correspond to two points, say $r$ and $s$, in the set $S_{g,i}$. 
We claim that $r$ and $s$ form the closest pair in the set 
$\llbracket H_{\lfloor r_1 \rfloor} , H_{\lceil s_1 \rceil} \rrbracket \cap S$, where $r_1$ and $s_1$ are
the first coordinates of $r$ and $s$, respectively (assuming that 
$r_1 < s_1$). Note that we take the floor and the ceiling, because $r_1$ and $s_1$ are not integers. 

This claim will imply that 
\[ f_2(n,m) \geq |A(S,m)| \geq 
  \sum_{g=1}^{n/m} \sum_{i=1}^{2g-1} k_g \cdot 2^{k_g-1} .
\]
Since 
$k_g > \log \left( \frac{m}{2g-1} \right) - 1$, 
we get 
\begin{eqnarray*}
    f_2(n,m) & \geq & 
    \sum_{g=1}^{n/m} (2g-1) 
       \left( \log \left( \frac{m}{2g-1} \right) - 1 \right) \cdot 
          \frac{m}{4(2g-1)} \\ 
    & = & \sum_{g=1}^{n/m} \frac{m}{4} 
                    \log \left( \frac{m}{2g-1} \right) - 
                    \sum_{g=1}^{n/m} \frac{m}{4} . 
\end{eqnarray*}
Since each term in the first summation is larger than the last term, 
which is larger than $(m/4) \log ( m^2 / (2n) )$, we get 
\begin{eqnarray*} 
 f_2(n,m) & \geq & \frac{n}{m} \cdot \frac{m}{4} \log ( m^2 / (2n) ) 
          - \frac{n}{4} = \frac{n}{4} 
          \left( \log (m^2/(2n)) - 1 \right) . 
\end{eqnarray*}  
Since $m > 3 \sqrt{n}$, 
\[ \log (m^2/(2n)) - 1 = \Omega (\log (m^2/n)) . 
\]
We conclude that 
\[ f_2(n,m) = \Omega (n \log (m^2/n)) . 
\]

\paragraph{Proof of the main claim:}
Let $g$ and $i$ be integers with $1 \leq g \leq n/m$ and 
$1 \leq i \leq 2g-1$. Consider two binary strings of length $k_g$ that 
differ in exactly one position; denote this position by $t$. 
Let $r$ and $s$ be the two corresponding points of $S_{g,i}$. 
Note that the vector $v_{g,t}$ is equal to either $r-s$ or $s-r$. 

We may assume, without loss of generality, that $r_1 < s_1$.  
To prove that $r$ and $s$ form the closest pair in the set 
$\llbracket H_{\lfloor r_1 \rfloor} , H_{\lceil s_1 \rceil} \rrbracket \cap S$, we consider an 
arbitrary pair $p$ and $q$ of points in 
$\llbracket H_{\lfloor r_1 \rfloor} , H_{\lceil s_1 \rceil} \rrbracket \cap S$ such that 
$\{p,q\} \neq \{r,s\}$. We will show that $|| r-s || < || p-q ||$.

Let $g'$, $g''$, $i'$, and $i''$ be such that 
$p \in S_{g',i'}$ and $q \in S_{g'',i''}$. If $g' \neq g''$ or 
$i' \neq i''$, then 
\[ \| p-q \| \geq | p_2 - q_2 | \geq \Delta > 
     \DIAM(Q_g) \geq \| r-s \| . 
\]
In the rest of the proof, we assume that $g' = g''$ and $i' = i''$. 
Since both $p$ and $q$ are in $S_{g',i'}$, we can write them as 
\[ p = \sum_{j=0}^{k_{g'}-1} \beta_{p,j} v_{g',j} + 
   ( i' + \eps_{g'} , 2 ( (g'-1)^2 + i'-1 ) \Delta ) 
\]
and 
\[ q = \sum_{j=0}^{k_{g'}-1} \beta_{q,j} v_{g',j} + 
   ( i' + \eps_{g'} , 2 ( (g'-1)^2 + i'-1 ) \Delta ) . 
\]
Let $\ell$ be the smallest index such that $0 \leq \ell \leq k_{g'}-1$ 
and $\beta_{p,\ell} \neq \beta_{q,\ell}$. Using the triangle 
inequality and (\ref{eq2}), we have 
\[
 || p-q ||  =  
 \left| \left| \sum_{j=\ell}^{k_{g'}-1} ( \beta_{p,j} - \beta_{q,j} ) 
                            v_{g',j} 
 \right| \right| 
  \geq  || v_{g',\ell} || - 
          \sum_{j=\ell+1}^{k_{g'}-1} || v_{g',j}|| 
     >  \frac{1}{2} \cdot || v_{g',\ell} || . 
\]

Thus, it suffices to show that 
\[ || v_{g',\ell} || \geq 2 \cdot || r-s || = 2 \cdot || v_{g,t} || . 
\]
Since both $p$ and $q$ are in $S_{g',i'}$, we have $|p_1 - q_1| \le |r_1-s_1|$ (where $r$ and $s$ are each at distance $\eps_g$ from the left boundary of their corresponding slabs, and $p$ and $q$ are at distance $\eps_{g'}$ from the left boundary of their corresponding slabs). Let $\ell'$ be the largest index such that $0 \leq \ell' \leq k_{g'}-1$ 
and $\beta_{p,\ell'} \neq \beta_{q,\ell'}$. We have 
\begin{eqnarray*}
    x_{g,t} & = & | r_1 - s_1| \\ 
    & \geq & | p_1 - q_1 | \\
    & = & \left| \sum_{j=\ell}^{\ell'} ( \beta_{p,j} - \beta_{q,j} ) 
                       x_{g',j} 
           \right| \\ 
    & \geq & x_{g',\ell'} - \sum_{j=\ell}^{\ell'-1} x_{g',j} \\ 
    & = & ( 2g'-1 ) 
    \left( 2^{\ell'} - \sum_{j=\ell}^{\ell'-1} 2^j
    \right) \\ 
     & = & ( 2g'-1 ) \cdot 2^{\ell} \\ 
     & = & x_{g',\ell} . 
\end{eqnarray*}
Since $\{p,q\} \neq \{r,s\}$, $x_{g,t}$ cannot be equal to $x_{g',\ell}$. 
Therefore, $x_{g,t} > x_{g',\ell}$. Using (\ref{eq2}), it then follows that 
$|| v_{g',\ell} || \geq 2 \cdot || v_{g,t} ||$. 
\end{proof}

\section{The Data Structure} \label{data-structure}

In this section, we will present a data structure for vertical closest
pair queries. Our data structure will use the results in the following three lemmas. 

\begin{lem} 
\label{lemshortest}
Let $S$ be a set of $n$ points in $\R^d$, and let $L$ be a set of $k$ line segments such that the endpoints of each segment belong to $S$. 
There exists a data structure of size $O(k)$, such that for any two 
real numbers $a$ and $b$ with $a<b$, the shortest segment in $L$ that 
is completely inside the vertical slab $\llbracket H_a,H_b \rrbracket$
can be reported in $O(\log n)$ time.
\end{lem}
The proof follows from 
Xue \emph{et al.}~\cite[Section 3]{DBLP:journals/dcg/XueLRJ22}, who proved 
a similar result for the case when $d=2$. For completeness, we include a 
proof of the lemma here.
\begin{proof}
Consider a line segment $pq$ in $L$, where 
$p=(p_1,p_2,\ldots,p_d)$, $q=(q_1,q_2,\ldots,q_d)$, and $p_1 \leq q_1$. 
We map this segment to the north-west quadrant 
\[ \NW(p_1,q_1) = \{ (x_1,x_2) \in \R^2 : x_1 \leq p_1 \mbox{ and }
               x_2 \geq q_1 \} . 
\]
Observe that the segment $pq$ is in the vertical slab 
$\llbracket H_a,H_b \rrbracket$ if and only if the point $(a,b)$ is in 
$\NW(p_1,q_1)$.

Let $Q = \{ \NW(p_1,q_1) : pq \in L \}$. 
We assign each quadrant $\NW(p_1,q_1)$ a weight which is equal to the distance 
$|| p-q ||$ between the points $p$ and $q$. Finding 
the shortest segment in $L$ that 
is completely inside the vertical slab $\llbracket H_a,H_b \rrbracket$ 
is equivalent to finding the quadrant in $Q$ of minimum weight that contains
the point $(a,b)$.

We lift each quadrant $\NW(p_1,q_1)$ in $Q$ to $\R^3$ 
by giving it $||p-q||$ as its third coordinate. Consider the lower envelope 
of these lifted quadrants. This is a planar subdivision of size 
$O(k)$. Finding the shortest segment inside $\llbracket H_a,H_b \rrbracket$
is equivalent to locating the point $(a,b)$ in this planar subdivision. 
For a planar subdivision of size $O(k)$, there are point location data structures 
of $O(k)$ size 
and $O(\log k)$ query time~\cite{Snoeyink17}. The query time is $O(\log n)$ because $k \leq {n \choose 2}$.
\end{proof}

The next lemma is due to Mehlhorn~\cite[page~44]{m-mdscg-84}; see Smid~\cite{smid} for a complete analysis of this data structure.

\begin{lem} \label{lem:range-report}
Let $S$ be a set of $n$ points in $\R^d$ and let $\eps>0$ be a real constant. There exists a data structure of size $O(n)$, such that for any axis-parallel rectangular box $B$, all points in $B \cap S$ can be reported in $O(n^\eps + |B \cap S|)$ time. 
\end{lem}

The last tool that we need is a standard sparsity property. 

\begin{lem} \label{obs:constant-points}
    Let $r>0$ be a real number, and let $X$ be a set of points in $\R^d$ that are contained in an $r \times 2r \times 2r \times \cdots \times 2r$ rectangular box $B$. If the distance of the closest pair of points in $X$ is at least $r$, then $|X| \le 2^{d-1} \cdot c^d$, where $c = 1 + \lceil \sqrt{d} \rceil$.
\end{lem}
\begin{proof}
Partition $B$ into $2^{d-1} \cdot c^d$ hypercubes, each one being an
$\frac{r}{c} \times\frac{r}{c}  \times \cdots \times \frac{r}{c}$ 
box. Since the diameter of any such box is $\sqrt{d} \cdot r/c$, which is 
less than $r$, it contains at most one point of $X$.
\end{proof}

In the rest of this section, we will prove the following result.

\begin{thm} \label{thm:main2}
    Let $S$ be a set of $n$ points in $\R^d$, let $m$ be an integer 
    with $1 \leq m \leq n$, and let $\eps>0$ be a real constant. 
    There exists a data structure of size $O(n + f_d(n,m))$ such that 
    for any two real numbers $a$ and $b$ with $a<b$, the closest pair 
    in the vertical slab $\llbracket H_a,H_b \rrbracket$ can be 
    reported in $O(n^{1+\eps} /m)$ time.
\end{thm}

\begin{proof}
Let $a_1 < a_2 < \cdots < a_{m+1}$ be real numbers such that for each 
$i = 1,2,\ldots,m$, the vertical slab 
$\llbracket H_{a_i} , H_{a_{i+1}} \rrbracket$ contains $n/m$ points 
of $S$. Let $k = |A(S,m)|$. Note that $k \leq f_d(n,m)$. Our data structure 
consists of the following components: 

\begin{itemize} 
\item An array storing the numbers $a_1,a_2,\ldots,a_{m+1}$. For each 
$i=1,2,\ldots,m$, the $i$-th entry stores, besides the number $a_i$, 
a list of all points in 
$\llbracket H_{a_i} , H_{a_{i+1}} \rrbracket \cap S$. 
\item The data structure of Lemma~\ref{lemshortest}, where 
$L$ is the set of line segments corresponding to the pairs in 
$\{ \CP(S,H_{a_i},H_{a_j}) : 1 \leq i < j \leq m+1 \}$. 
\item The data structure of Lemma~\ref{lem:range-report}.
\end{itemize} 
The size of the entire data structure is 
$O(m+n+k)$, which is $O(n + f_d(n,m))$. 

We next describe the query algorithm. Let $a$ and $b$ be real numbers 
with $a<b$. Using binary search, we compute, in 
$O(\log m) = O(n^{1+\eps} / m)$ time, 
the indices $i$ and $j$ such that $H_a$ is in the slab 
$\llbracket H_{a_{i-1}} , H_{a_i} \rrbracket$ 
and $H_b$ is in the slab 
$\llbracket H_{a_j} , H_{a_{j+1}} \rrbracket$. 

If $i=j$, then the slab $\llbracket H_a,H_b \rrbracket$ contains 
$O(n/m)$ points of $S$. In this case, we use the  
algorithm of Bentley and Shamos~\cite{bs-dcms-76} to compute the 
closest pair in $\llbracket H_a,H_b \rrbracket$ in 
$O((n/m) \log (n/m)) = O(n^{1+\eps} / m)$ time.

\begin{figure}
    \begin{center}
        \includegraphics[scale=0.75]{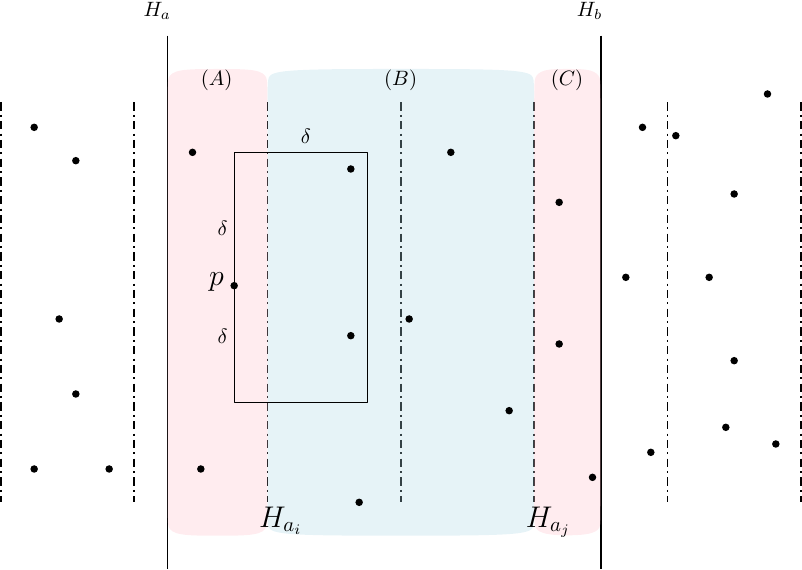}
    \end{center}
    \caption{(A), (B), and (C) are the three regions created by a query $\llbracket H_a,H_b \rrbracket$. The rectangle $R_p$ is the range query for the point $p$ with respect to the query $\llbracket H_a,H_b \rrbracket$.}
    \label{figs:regions}
\end{figure}

Assume that $i < j$. The two hyperplanes $H_{a_i}$ and $H_{a_j}$ split the query slab $\llbracket H_a, H_b \rrbracket$ into three parts $(A)$, $(B)$, and $(C)$, where 
$(A)$ is the slab $\llbracket H_a , H_{a_i} \rrbracket$, 
$(B)$ is the slab $\llbracket H_{a_i} , H_{a_j} \rrbracket$, and 
$(C)$ is the slab $\llbracket H_{a_j} , H_b \rrbracket$; refer to 
Figure~\ref{figs:regions}. 

Let $S_{AC}$ be the set of points in $S$ that are in the union of $(A)$ 
and $(C)$, and let $S_B$ be the set of points in $S$ that are in $(B)$. 
There are three possibilities for the closest pair in 
$\llbracket H_a,H_b \rrbracket$: Both endpoints are in $S_{AC}$, both 
endpoints are in $S_B$, or one endpoint is in $S_{AC}$ and the other 
endpoint is in $S_B$. 

Using the algorithm of Bentley and Shamos~\cite{bs-dcms-76}, we 
compute the closest pair distance $\delta_1$ in $S_{AC}$, in 
$O((n/m) \log (n/m)) = O(n^{1+\eps} / m)$ time. Using the data 
structure of Lemma~\ref{lemshortest}, we compute the closest pair 
distance $\delta_2$ in $S_B$ in 
$O(\log n) = O(n^{1+\eps} / m)$ time. 

Let $\delta = \min(\delta_1,\delta_2)$. In the final part of the query 
algorithm, we use the data structure of Lemma~\ref{lem:range-report}:

\begin{itemize} 
\item For each point $p$ in the region $(A)$, we compute the set of 
all points in $S$ that are in the part, say $P_p$, of the axes-parallel box 
\[ [ p_1,p_1+\delta ] \times [ p_2-\delta, p_2+\delta]
    \times \cdots \times [ p_d-\delta,p_d+\delta] 
\]
that is to the left of $H_{a_j}$. Then we compute $\delta_p$, which is 
the minimum distance between $p$ and any point in
$(S \cap P_p) \setminus \{p\}$. 
\item For each point $p$ in the region $(C)$, we compute the set of 
all points in $S$ that are in the part, say $P'_p$, of the axes-parallel box 
\[ [ p_1-\delta,p_1 ] \times [ p_2-\delta, p_2+\delta]
    \times \cdots \times [ p_d-\delta,p_d+\delta] 
\]
that is to the right of $H_{a_i}$. Then we compute $\delta_p$, which is
the minimum distance between $p$ and any point in 
$(S \cap P'_p) \setminus \{p\}$. 
\item At the end, we return the minimum of $\delta$ and 
$\min \{ \delta_p : p \in S_{AC} \}$. 
\end{itemize}
By Lemma~\ref{obs:constant-points}, the boxes $P_p$ and $P'_p$ each contain
$O(1)$ points of $S$. In total, there are $O(n/m)$ queries to the 
data structure of Lemma~\ref{lem:range-report}, and each one 
takes $O(n^{\eps})$ time. Thus, this final part of the query 
algorithm takes $O((n/m) \cdot n^{\eps}) = O(n^{1+\eps} / m)$ time. 
\end{proof}

The proof of Theorem~\ref{thm-main2} follows by taking $m = \sqrt{n}$ 
in Theorem~\ref{thm:main2} and using Theorem~\ref{thm-main1}.



\section{Future Work}
\label{future} 

The point sets that we constructed for the lower bounds on 
$f_d(n,m)$ have coordinates that are at least 
exponential 
in the number of points. Recall that the \emph{spread} 
(also known as \emph{aspect ratio})
of a point set is the ratio of the diameter and the closest pair distance. 
It is well-known that the spread of any set of $n$ points in $\R^d$ is 
$\Omega(n^{1-1/d})$. It is natural to define 
$f_d(n,m,\Phi)$ as the quantity analogous to $f_d(n,m)$, where 
we only consider sets of $n$ points in $\R^d$ having spread at most 
$\Phi$. 

\begin{prb}
    Determine the value of $f_d(n,m,\Phi)$.
\end{prb}

For any set $S$ of $n$ points in $\R^d$, where $d=2$, 
Xue \emph{et al.}~\cite{DBLP:journals/dcg/XueLRJ22} have presented a 
data structure of size $O(n \log n)$ that can be used to answer 
vertical slab closest pair queries in $O(\log n)$ time. Our data 
structure uses only $O(n)$ space and works in any constant dimension 
$d \geq 2$. However, its query time is $O(n^{1/2+\eps})$. 

\begin{prb}
    Is there a linear space data structure that supports vertical slab closest pair queries in $o(\sqrt{n})$ time, or even 
    in $O(\polylog{\!(n)})$ time?
\end{prb}


Another interesting research direction is to design linear space data 
structures for closest pair queries with other types of ranges, such 
as axes-parallel 3-sided ranges and general axes-parallel 
rectangular ranges. 

\section*{Acknowledgements}
Research on this paper was conducted, in part, at the Eleventh Annual 
Workshop on Geometry and Graphs, held at the Bellairs Research Institute 
in Barbados, March 8--15, 2024. The authors are grateful to the 
organizers and to the participants of this workshop. 
The authors thank the anonymous referees for their 
detailed comments on a preliminary version of this paper.

\bibliographystyle{plainurl}
\bibliography{cpq}

\end{document}